\documentclass[runningheads]{llncs}
\usepackage[T1]{fontenc}
\usepackage{graphicx}
\usepackage{hyperref}
\usepackage{color}

\usepackage[table,x11names]{xcolor}
\usepackage{amsmath}
\usepackage{amsfonts}
\usepackage{mathtools}
\usepackage{amssymb}
\usepackage{enumerate}
\usepackage{longtable}
\usepackage{float}
\usepackage{babel}
\usepackage{orcidlink}

\usepackage{adjustbox}
\usepackage{array}

\usepackage{hyperref}
\hypersetup{
	colorlinks   = true, 
	urlcolor     = red, 
	linkcolor    = blue, 
	citecolor   = blue 
}

\DeclareMathOperator{\hada}{Hada} 
\DeclareMathOperator{\diag}{diag} 

\usepackage[
n,
operators,
advantage,
sets,
adversary,
landau,
probability,
notions,	
logic,
ff,
mm,
primitives,
events,
complexity,
asymptotics,
keys]{cryptocode}
\newcommand{\splitatcommas}[1]{%
	\begingroup
	\begingroup\lccode`~=`, \lowercase{\endgroup
		\edef~{\mathchar\the\mathcode`, \penalty0 \noexpand\hspace{0pt plus 1em}}%
	}\mathcode`,="8000 #1%
	\endgroup
}
\newtheorem{fact}{\bf Fact}

\begin{document}
	
\mainmatter          

\title{On the Counting of Involutory MDS Matrices}
%

\author{Susanta Samanta~\thanks{Part of this work was conducted during the author's affiliation with the Applied Statistics Unit, Indian Statistical Institute, India.}}


\authorrunning{~}

\institute{
	Department of Electrical and Computer Engineering\\
	University of Waterloo\\
	Waterloo, Ontario, N2L 3G1, Canada \\ 
	\email{ssamanta@uwaterloo.ca}
}
\maketitle

\begin{abstract}
    \sloppy
	The optimal branch number of MDS matrices has established their importance in designing diffusion layers for various block ciphers and hash functions. As a result, numerous matrix structures, including Hadamard and circulant matrices, have been proposed for constructing MDS matrices. Also, in the literature, significant attention is typically given to identifying MDS candidates with optimal implementations or proposing new constructions across different orders. However, this paper takes a different approach by not emphasizing efficiency issues or introducing new constructions. Instead, its primary objective is to enumerate Hadamard MDS and involutory Hadamard MDS matrices of order $4$ within the field $\mathbb{F}_{2^r}$. Specifically, it provides an explicit formula for the count of both Hadamard MDS and involutory Hadamard MDS matrices of order $4$ over $\mathbb{F}_{2^r}$. Additionally, it derives the count of Hadamard Near-MDS (NMDS) and involutory Hadamard NMDS matrices, each with exactly one zero in each row, of order $4$ over $\mathbb{F}_{2^r}$. Furthermore, the paper discusses some circulant-like matrices for constructing NMDS matrices and proves that when $n$ is even, any $2n \times 2n$ Type-II circulant-like matrix can never be an NMDS matrix. While it is known that NMDS matrices may be singular, this paper establishes that singular Hadamard matrices can never be NMDS matrices. Moreover, it proves that there exist exactly two orthogonal Type-I circulant-like matrices of order $4$ over $\mathbb{F}_{2^r}$.

\end{abstract}


\section{Introduction}
Claude Shannon, in his paper ``Communication Theory of Secrecy Systems'' \cite{SHANON}, introduced the concepts of confusion and diffusion, which play a significant role in the design of symmetric key cryptographic primitives.
The concept of confusion aims to create a statistical relationship between the ciphertext and message that is too intricate for an attacker to exploit. This is accomplished through the use of nonlinear functions such as Sboxes and Boolean functions. Diffusion, on the other hand, ensures that each bit of the message and secret key influences a significant number of bits in the ciphertext, and over several rounds, all output bits depend on every input bit.

\sloppy Optimal diffusion layers can be achieved by employing \textit{MDS matrices} with the highest branch numbers. As a result, various matrix structures have been suggested for the designing of MDS matrices, including \textit{Hadamard} and \textit{circulant} matrices. A concise survey on the various theories on the construction of MDS matrices is provided in~\cite{MDS_Survey}. In the context of lightweight cryptographic primitives, the adoption of \textit{involutory matrices} allows for the implementation of both encryption and decryption operations using identical circuitry, thereby resulting in an equivalent implementation cost for both processes. So, it is of special interest to find efficient MDS matrices which are also involutory.

Generally, three techniques are used to generate involutory MDS matrices. The first technique, known as direct construction, uses algebraic methods to produce an involutory MDS matrix of any order without requiring a search process. The second approach relies on search-based methods, while the third combines these two strategies into a hybrid methodology. Direct constructions are typically derived from Cauchy and Vandermonde matrices, along with their extensions. Although direct construction methods enable the generation of MDS matrices of any order, they do not guarantee matrices with optimal hardware efficiency, even for smaller sizes. In contrast, the search-based technique is currently the only approach that can yield an MDS matrix with the best possible hardware area. However, this method is feasible only when the field size and matrix size remain small. A concise overview of both direct and search-based constructions of MDS matrices, including their involutory properties, can be found in~\cite{MDS_Survey}.


In the hybrid approach, a representative involutory MDS matrix is generally obtained by search. This matrix then serves as a basis to generate multiple involutory MDS matrices. The concept of a representative matrix for involutory MDS matrices first appeared in \cite{GUZEL2019}, where a form was presented that generates all $3 \times 3$ representative involutory MDS matrices using two parameters. However, it is noteworthy that this paper did not apply a hybrid method with representative matrices; instead, it provided an explicit form of a $3 \times 3$ involutory MDS matrix over $\mathbb{F}_{2^r}$. From this form, all $3 \times 3$ involutory MDS matrices can be derived by varying parameters over $\mathbb{F}_{2^r}$ under certain conditions. The authors further demonstrated that there are $(2^r-1)^2 (2^r-2)(2^r-4)$ involutory MDS matrices of size $3 \times 3$ over the finite field $\mathbb{F}_{2^r}$. In \cite{GHADA}, the authors introduced a hybrid method to efficiently generate involutory MDS matrices of size $2^n \times 2^n$. This approach introduced the GHadamard matrix, a new matrix form that enables the generation of new involutory MDS matrices using an involutory Hadamard MDS matrix as a representative matrix. In \cite{SakalliAutomorphism}, a technique was presented for producing new $n \times n$ involutory MDS matrices isomorphic to pre-existing ones.

In a recent work \cite{Kumar_MDS2024}, the authors demonstrated a hybrid approach to generate all $n \times n$ involutory MDS matrices over $\mathbb{F}_{2^r}$, resulting in a search space of $2^{(n-1)^2r}$. Other recent studies, such as \cite{Kumar_4MDS} and \cite{Tuncay23}, proposed hybrid methods for $4 \times 4$ involutory MDS matrices over $\mathbb{F}_{2^r}$. In \cite{Tuncay23}, the search space for finding representative involutory MDS matrices of order 4 was reduced significantly from $2^{16r}$ to $2^{8r}$. Subsequently, \cite{Kumar_4MDS} further optimized this search space to $2^{5r}$ for finding representative involutory MDS matrices of order 4 over $\mathbb{F}_{2^r}$. Then, by exhaustive search in the reduced search space, the authors provided the count of $4 \times 4$ involutory MDS matrices over $\mathbb{F}_{2^r}$ for $r = 3, 4, \ldots, 8$. In addition to the hybrid approach, there are studies focusing on the hardware area and latency of $4 \times 4$ involutory MDS matrices. Notable examples from the literature include references~\cite{Bai2020,Li2019FSE,Meltem2023,Shi_2023AMC,Yang_Zeng_Wang_2021,Zhao2024}. However, there is currently no exact form of a $4 \times 4$ matrix that offers a direct construction method for generating all involutory MDS matrices of order $4$ over any finite field of characteristic 2, as exists for $3\times 3$ involutory matrices in~\cite{GUZEL2019}.

One of the most noteworthy advantages of Hadamard matrices lies in their capability to facilitate the construction of involutory matrices. If the matrix elements are selected such that the first row sums to one, the resultant matrix attains involutory properties~\cite{MDS_Survey}. Due to this advantageous characteristic, several block ciphers, such as Anubis ~\cite{ANUBIS}, Khazad~\cite{KHAZAD} and CLEFIA ~\cite{CLEFIA}, have incorporated involutory Hadamard MDS matrices into their diffusion layers.

It is noteworthy that in the literature, most work often prioritizes identifying MDS candidates with optimal implementations or provides direct (theoretical) constructions of MDS matrices of any order. Examples from the literature include references~\cite{Gupta2023direct,LACAN2003,CYCLICM,GHADA,ROTH2,V_MDS,PXOR1}. However, this paper does not focus on efficiency issues nor propose new constructions. Instead, its primary aim is to enumerate Hadamard MDS and involutory Hadamard MDS matrices of order $4$ within the field $\mathbb{F}_{2^r}$. Specifically, it proves that the counts of $4\times 4$ Hadamard MDS and $4\times 4$ involutory Hadamard MDS matrices over $\mathbb{F}_{2^r}$ are $(2^r-1)(2^r-2)(2^r-4)(2^r-7)$ and $(2^r-2)(2^r-4)(2^r-7)$ respectively. Additionally, it derives a general form for $2n\times 2n$ involutory MDS matrices and using this it provides counts for both order $2$ MDS matrices and order $2$ involutory MDS matrices in the field $\mathbb{F}_{2^r}$. Finally, leveraging these counts of order $2$ matrices, the paper establishes an upper limit for the number of all involutory MDS matrices of order $4$ over $\mathbb{F}_{2^r}$.

In~\cite{Combo_NMDS_2023}, the authors studied various matrix structures, including Hadamard, circulant, and left-circulant matrices, in the context of constructing \textit{Near-MDS (NMDS) matrices}. However, some circulant-like matrices have yet to be analyzed in relation to NMDS matrices. Additionally, \cite{GR15} demonstrated that when $n$ is even, the corresponding $2n \times 2n$ Type-II circulant-like matrix cannot be MDS. However, whether this holds true for NMDS matrices remains unknown. This paper aims to address these gaps. Specifically, in Theorem~\ref{Th_TypeII_NMDS}, we prove that for even $n$, a Type-II circulant-like matrix of order $2n$ can never be an NMDS matrix. We also show that there are exactly two orthogonal Type-I circulant-like matrices of order $4$ over $\mathbb{F}_{2^r}$. Furthermore, as established in~\cite{Combo_NMDS_2023}, an NMDS matrix can be singular. In this context, Theorem~\ref{Th_Hadamard_singular} demonstrates that a singular Hadamard matrix can never be an NMDS matrix over $\mathbb{F}_{2^r}$. Finally, we prove that the number of $4 \times 4$ Hadamard NMDS and involutory Hadamard NMDS matrices $\hada(a_1, a_2, a_3, a_4)$, with exactly one $a_i = 0$, over $\mathbb{F}_{2^r}$ is given by $4(2^r - 1)(2^{2r}-3\cdot 2^r+3)$ and $4(2^{2r}-3\cdot 2^r+3)$, respectively.
\vspace{1em}

\noindent The structure of the paper is as follows. In Section~\ref{Sec:Definition}, we provide an overview of the mathematical background and notations used throughout the paper. Section~\ref{Sec:4Hadamard_MDS} is dedicated to deriving the count of $4\times 4$ Hadamard MDS and involutory Hadamard MDS matrices. In Section~\ref{Sec:general_2n_MDS}, we present a general form for $2n\times 2n$ involutory MDS matrices and use it to derive the count of involutory MDS matrices of order 2. Section~\ref{Sec:Near-MDS} investigates NMDS matrices by leveraging Hadamard and circulant-like matrices. Finally, Section~\ref{Sec:Conclusion} concludes the paper.

\section{Definition and Preliminaries}\label{Sec:Definition}
Let $\mathbb{F}_2=\{0,1\}$ be the finite field of two
elements, $\mathbb{F}_{2^r}$ be the finite field of $2^r$ elements and $\FF_{2^r}^\star$ be the multiplicative group of $\FF_{2^r}$. The set of vectors of length $n$ with entries from the finite field $\mathbb{F}_{2^r}$ is denoted by $\mathbb{F}_{2^r}^n$.

A matrix $D$ of order $n$ is said to be diagonal if $(D)_{i,j}=0$ for $i\neq j$. Using the notation $d_i = (D)_{i,i}$, the diagonal matrix $D$ can be represented as $\diag(d_1, d_2, \ldots, d_n)$. It is evident that the determinant of $D$ is given by $\det(D) = \prod_{i=1}^{n} d_i$. Therefore, the diagonal matrix $D$ is nonsingular over $\mathbb{F}_{2^r}$ if and only if $d_i \neq 0$ for $1 \leq i \leq n$.


An MDS matrix provides diffusion properties that have practical applications in the field of cryptography. This concept originates from coding theory, specifically from the realm of maximum distance separable (MDS) codes. A $[n, k]$ code with minimum distance $d$ is considered MDS if it meets the Singleton bound, $d = n - k + 1$. If the code and its dual satisfy $d = n - k$, we refer to the code as a Near-MDS (NMDS) code.

\begin{theorem}~\cite[page 321]{FJ77}
    An $[n, k]$ code $C$ with generator matrix $G = [ I_k| M ]$, where $M$ is a $k \times ( n - k )$ matrix, is MDS if and only if every square submatrix (formed from any $i$ rows and any $i$ columns, for any $i = 1, 2 , \ldots, min \{k, n - k \}$) of $M$ is nonsingular.
\end{theorem}

\begin{definition}
	A matrix $M$ of order $n$ is said to be an MDS matrix if $[I_n|M]$ is a generator matrix of a $[2n,n]$ MDS code.
\end{definition}

Another way to define an MDS matrix is as follows.

\begin{fact}
    A square matrix $M$ is an MDS matrix if and only if every square submatrices of $M$ is nonsingular. 
\end{fact}

In lightweight cryptography, Near-MDS (NMDS) matrices offer a better balance between security and efficiency as a diffusion layer, compared to MDS matrices.

\begin{definition}
	A matrix $M$ of order $n$ is said to be a Near-MDS (NMDS) matrix if $[I_n|M]$ is a generator matrix of a $[2n,n]$ Near-MDS (NMDS) code.
\end{definition}

Similar to MDS matrices, we have the following characterization of an NMDS matrix.

\begin{lemma}\cite{Li_Wang_2017,Viswanath2006_nMMDS}\label{Lemma_nMDS_charac}
    Let $M$ be a non-MDS matrix of order $n$, where $n$ is a positive integer with $n \geq 2$. Then $M$ is NMDS if and only if for any $1 \leq g \leq n-1$ each $g \times (g + 1)$ and $(g + 1) \times g$ submatrix of $M$ has at least one $g \times g$ nonsingular submatrix.
\end{lemma}


\noindent One of the elementary row operations on matrices is multiplying a row of a matrix by a nonzero scalar. MDS (or NMDS) property remains invariant under such operations. Thus, we have the following result regarding MDS (or NMDS) matrices.

\begin{lemma}\label{Lemma_DMD_MDS}
	Let $M$ be an MDS (or NMDS) matrix, then for any nonsingular diagonal matrices $D_1$ and $D_2$, $D_1MD_2$ will also be an MDS (or NMDS) matrix.
\end{lemma}

\noindent Using involutory diffusion matrices is more beneficial for implementation since it allows the same module to be utilized in both encryption and decryption phases.

\begin{definition}
	An involutory matrix is defined as a matrix $M$ of order $n$ that satisfies the condition $M^2 = I_n$ or, equivalently, $M = M^{-1}$.
\end{definition}

Therefore, based on Lemma~\ref{Lemma_DMD_MDS}, we can deduce the following result.

\begin{corollary}\label{Corollary_DMD_MDS}
	For any nonsingular diagonal matrix $D$, the matrix $DMD^{-1}$ is an involutory MDS (or NMDS) matrix if and only if $M$ is also an involutory MDS (or NMDS) matrix.
\end{corollary}

Permuting the rows or columns of a matrix does not change the MDS (or NMDS) property. Thus, we have the following result regarding MDS (or NMDS) matrices.

\begin{lemma}\label{Lemma_PMQ_MDS}
	Let $M$ be an MDS (or NMDS) matrix, then for any two permutation matrices $P$ and $Q$, $PMQ$ will also be an MDS (or NMDS) matrix.
\end{lemma}




The search method for finding MDS (or NMDS) matrices is typically based on specific types of matrices, such as Hadamard matrices, circulant matrices, and their variants. Circulant and Hadamard matrices of order $n$ can have at most $n$ distinct elements; thus, these matrices are utilized in search method to reduce the search space.

\begin{definition}\label{CIRCULANT}
    An $n\times n$ matrix $M$ is said to be a circulant matrix if its elements are determined by the elements of its first row $a_1,a_2,\ldots,a_n$ as
    $$M
    =\begin{bmatrix}
    a_1 & a_2 & \ldots & a_{n}\\
    a_{n} & a_1 & \ldots & a_{n-1}\\
    \vdots & \vdots &\vdots &\vdots \\
    a_2 & a_3 & \ldots & a_1\\
    \end{bmatrix}.$$
\end{definition}

For simplicity, we will use $Circ(a_1, a_2, \ldots, a_n)$ to denote the circulant matrix with the first row given by $a_1, a_2, \ldots, a_n$.

\begin{definition}
	A matrix $M$ of size $2^n \times 2^n$ in the field $\mathbb{F}_{2^r}$ is called a Finite Field Hadamard matrix, or simply a Hadamard matrix, if it can be represented in the following form:
	\begin{center}
		M=
		$\begin{bmatrix}
		U & V \\
		V & U
		\end{bmatrix}$
	\end{center}
	where the submatrices $U$ and $V$ are also Hadamard matrices.
\end{definition}

For example, a $2^2\times 2^2$ Hadamard matrix is:
\begin{align*}
    M=\begin{bmatrix}
        a_1 & a_2 & a_3 & a_4\\
        a_2 & a_1 & a_4 & a_3\\
        a_3 & a_4 & a_1 & a_2 \\
        a_4 & a_3 & a_2 & a_1
    \end{bmatrix}.
\end{align*}

Note that Hadamard matrices are symmetric and can be represented by their first row. For simplicity, we will denote a Hadamard matrix with the first row as $a_1, a_2, \ldots, a_n$ as $\hada(a_1, a_2, \ldots, a_n)$. 

Also, it is worth noting that $H^2=c^2I_n$, where $c=a_1 + a_2 + \cdots + a_n$~\cite{MDS_Survey}. Thus, when $c=1$, the Hadamard matrix will be involutory.

\section{Enumeration of $4\times 4$ Hadamard MDS matrices}\label{Sec:4Hadamard_MDS}
In this section, we enumerate $4\times 4$ Hadamard MDS matrices, including involutory Hadamard MDS matrices, over the finite field $\mathbb{F}_{2^r}$. First, we present the conditions that the Hadamard matrix $\hada(a,b,c,d)$ must satisfy to be considered an MDS matrix.

\begin{lemma}\label{Lemma_Hadamard_MDS}
	The Hadamard matrix $\hada(a,b,c,d)$ over $\mathbb{F}_{2^r}$ is MDS if and only if the tuple $(a,b,c,d) \in \mathbb{F}_{2^r}^4$ satisfies the  conditions: (i)~$a,b,c,d\in \mathbb{F}_{2^r}^\star$, (ii)~$\set{a,b,c,d}$  being a set of four distinct elements, (iii)~$d\neq a^{-1}bc$, (iv)~$d\neq ab^{-1}c$, (v)~$d\neq abc^{-1}$, and (vi)~$d\neq a+b+c$.
\end{lemma}

\begin{proof}
	The set of minors of $M=\hada(a,b,c,d)$ is given by: 
	\begin{center}
		\begin{small}
			$\{\splitatcommas{a, b, c, d, a^2 + b^2, bc + ad, ac + bd, c^2 + d^2, a^2 + c^2, ab + cd, b^2 + d^2, a^2 + d^2, b^2 + c^2, a^3 + ab^2 + ac^2 + ad^2, a^2b + b^3 + bc^2 + bd^2, a^2c + b^2c + c^3 + cd^2, a^2d + b^2d + c^2d + d^3, a^4 + b^4 + c^4 + d^4}\}$.
		\end{small}
	\end{center}
	
	These minors have factors given by:
	\begin{small}
		\begin{equation}\label{Eqn_Factor_minors}
			\{ \splitatcommas{a, b, c, d, a + b, bc + ad, ac + bd, c+d, a+c, ab+cd, b+d, a+d, b+c, a+b+c+d} \}.
		\end{equation}
	\end{small}

	Therefore, $M$ is an MDS matrix if and only if each element in the set given in~\ref{Eqn_Factor_minors} is nonzero. This condition is satisfied if and only if $(a,b,c,d) \in \mathbb{F}_{2^r}^4$ satisfies the conditions: (i)~$a,b,c,d\in \mathbb{F}_{2^r}^\star$, (ii)~$\set{a,b,c,d}$  being a set of four distinct elements, (iii)~$d\neq a^{-1}bc$, (iv)~$d\neq ab^{-1}c$, (v)~$d\neq abc^{-1}$, and (vi)~$d\neq a+b+c$. This completes the proof. \qed
\end{proof}

\begin{theorem}\label{Th_Hadamard_MDS_Counting}
	For $r\geq 3$, the count of $4\times 4$ Hadamard MDS matrices over the finite field $\mathbb{F}_{2^r}$ is given by 
	\[(2^r-1)(2^r-2)(2^r-4)(2^r-7).\]
\end{theorem}

\begin{proof}
	According to Lemma~\ref{Lemma_Hadamard_MDS}, the number of $4\times 4$ Hadamard MDS matrices over $\mathbb{F}_{2^r}$ is equal to the cardinality of the set $S$, defined as:

	\begin{equation*}
		\begin{aligned}
			S &= &&\{(a,b,c,d)\in (\mathbb{F}_{2^r}^\star)^4:~\set{a,b,c,d}~ \text{being a set of four distinct elements}\\
			  & &&~~\text{and} ~\splitatcommas{d\neq a^{-1}bc,~d\neq ab^{-1}c,~d\neq abc^{-1},~d\neq a+b+c}\}.
		\end{aligned}
	\end{equation*}


	Since $a \in \mathbb{F}_{2^r}^\star$, there are $2^r-1$ possible choices for $a$. Furthermore, with $b \in \mathbb{F}_{2^r}^\star$ and $a \neq b$, there are $2^r-2$ possible choices for $b$. Similarly, for $c$, there are $2^r-3$ possible choices.

	From the above conditions on $d$, we can say that

	\begin{equation*}\label{Eqn_set_D}
		\begin{aligned}
			d & \not \in T=\set{0,a,b,~c,a^{-1}bc, ab^{-1}c,~ abc^{-1}, a+b+c}.
		\end{aligned}
	\end{equation*}

	However, the cardinality of the set $T$, $|T|$, is not equal to 8 for all $a, b, c \in \mathbb{F}_{2^r}^\star$. For example, when $c = a^2b^{-1}$ (which is different from both $a$ and $b$), we have $a^{-1}bc = a$. Thus, when $c = a^2b^{-1}$, we have $|T| \leq 7$.

	\vspace*{1em}

	Now, we will demonstrate that $a^{-1}bc \not \in \set{b,c,ab^{-1}c,abc^{-1},a+b+c}$ for any choice of $a,b$ and $c$.

	\vspace{4pt}
	\noindent 
	\textbf{Case 1:} $a^{-1}bc=b$. \newline
	In this case, $a^{-1}bc=b$, which implies $a=c$. However, this contradicts our assumptions.

	\vspace{4pt}
	\noindent 
	\textbf{Case 2:} $a^{-1}bc=c$. \newline
	In this case, $a^{-1}bc=c$, which implies $a=b$, and this is a contradiction.

	\vspace{4pt}
	\noindent 
	\textbf{Case 3:} $a^{-1}bc=ab^{-1}c$. \newline
	Now, 
	\begin{equation*}
		\begin{aligned}
			& a^{-1}bc =ab^{-1}c \\
			& \implies a =b ~[\text{Since characteristic of}~ \mathbb{F}_{2^r}~\text{is}~2],
		\end{aligned}
	\end{equation*}
	which is a contradiction. Similarly, for the following cases, we will have a contradiction.

	\vspace{4pt}
	\noindent 
	\textbf{Case 4:} $a^{-1}bc=abc^{-1}$. This implies that $a=c$.

	\vspace{4pt}
	\noindent 
	\textbf{Case 5:} $a^{-1}bc=a+b+c$. This implies that $a=c$ or $a=b$.

	\vspace{4pt}
	\noindent Therefore, we have $a^{-1}bc \not \in \set{b,c,ab^{-1}c,abc^{-1},a+b+c}$.


	\vspace*{1em}

	Similarly, we can show the following:
	\begin{enumerate}[(i)]
		\setlength\itemsep{1em}
		\item $ab^{-1}c \not \in \set{a,c,a^{-1}bc,abc^{-1},a+b+c}$ and when $c=b^2a^{-1}$ (which is not equal to $a$ and not equal to $b$), we have $ab^{-1}c=b$.
		
		\item $abc^{-1} \not \in \set{a,b,a^{-1}bc,ab^{-1}c,a+b+c}$. Since, the characteristic of $\mathbb{F}_{2^r}$ is 2, $x \mapsto x^2$ is an isomorphism over $\mathbb{F}_{2^r}$. Hence, there exist a unique element $\alpha \in \mathbb{F}_{2^r}$ such that $\alpha^2=ab$. Now for $c=\alpha$ (which is not equal to $a$ and not equal to $b$), we have $abc^{-1}=c$.
		
		\item $a+b+c \notin \set{a,b,c,a^{-1}bc,ab^{-1}c,abc^{-1}}$. However, $a+b+c$ may equal zero, and when $c=a+b$, we have $a+b+c=0$.
	\end{enumerate}

	\begin{figure}[h!]
        \includegraphics[width=\linewidth]{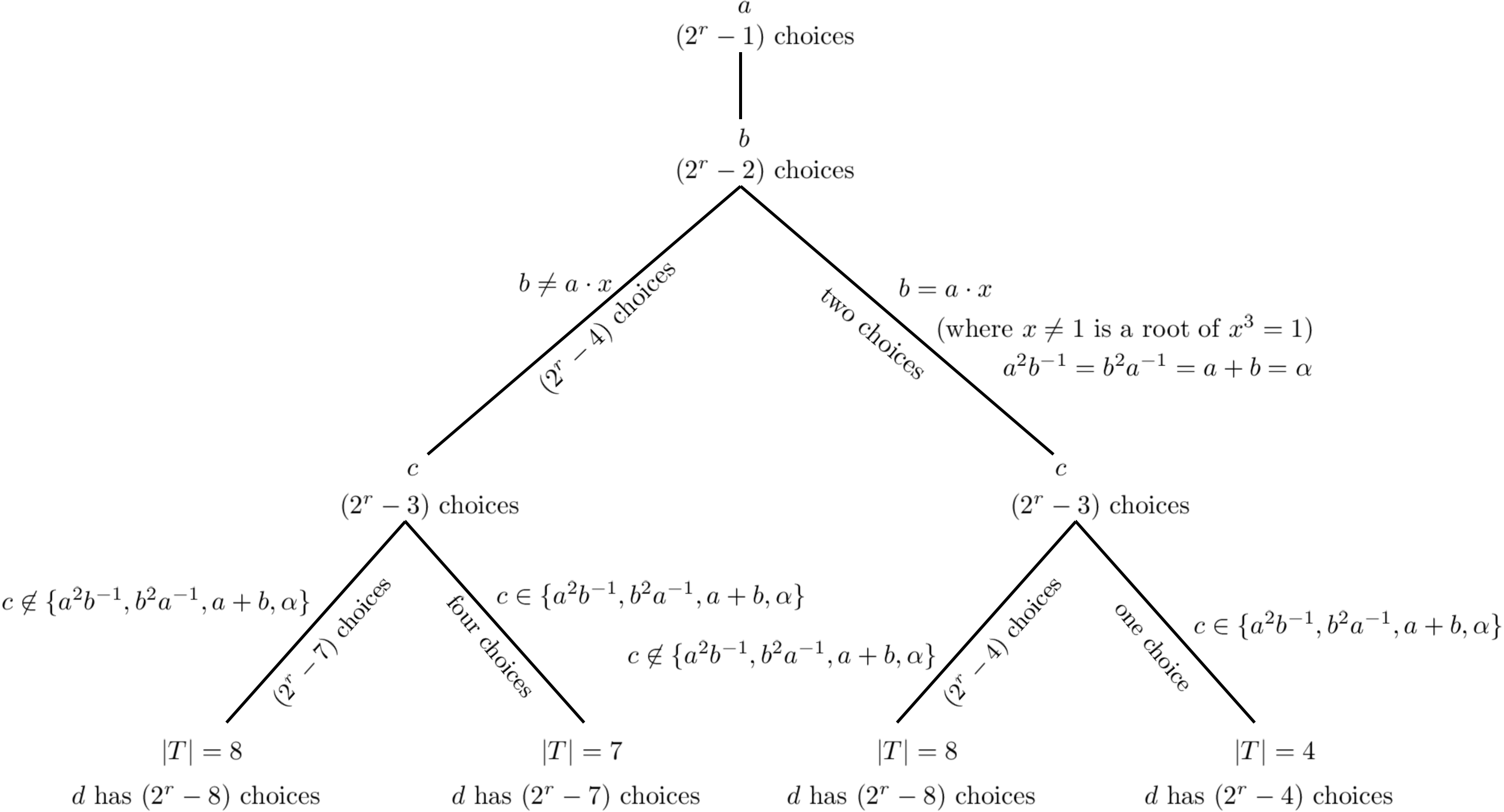}
        \caption{A figure illustrating the cases for determining the number of choices for d.}
		\label{Figure_choices_d}
    \end{figure}

	Therefore, when $c \in \{a^2b^{-1}, b^2a^{-1}, a+b, \alpha\}$, we have $|T| = 7$, so there are $2^r - 7$ choices for $d$. When $c \notin \{a^2b^{-1}, b^2a^{-1}, a+b, \alpha\}$, we have $|T| = 8$, so there are $2^r - 8$ choices for $d$.

	However, the elements in the set $\{a^2b^{-1}, b^2a^{-1}, a+b, \alpha\}$ may not always be distinct. For example, if we take $b = a \cdot x$, where $x \neq 1$ is a solution of the equation $x^3 = 1$, we have:

	\begin{center}
		$a^2b^{-1} = b^2a^{-1} = a + b = \alpha$.
	\end{center}

	Also, since $b\neq a$, we have only two such cases where $b=a\cdot x$ holds. Moreover, note that all the elements in the set $\{a^2b^{-1}, b^2a^{-1}, a+b, \alpha\}$ are distinct for any other choices of $a$ and $b$. Thus, we can divide the choices for $b$ into two cases:
	\begin{enumerate}
		\item When $b \neq a \cdot x$.
		\item When $b = a \cdot x$.
	\end{enumerate}

	Each case can be further divided into two subcases:
	\begin{enumerate}[(i)]
		\item $c \not \in \{a^2b^{-1}, b^2a^{-1}, a+b, \alpha\}$.
		\item $c \in \{a^2b^{-1}, b^2a^{-1}, a+b, \alpha\}$.
	\end{enumerate}

	Note that when $b = a \cdot x$ and $c \in \{a^2b^{-1}, b^2a^{-1}, a+b, \alpha\}$, we have $|T| = 4$. Hence, in this case, $d$ has $2^r - 4$ choices. For a detailed overview of all the cases, refer to Figure~\ref{Figure_choices_d}.


	Therefore, we have
	\begin{equation*}
		\begin{aligned}
			|S|& &&= (2^r-1)~[ (2^r-4)\cdot \set{(2^r-7)(2^r-8)+4\cdot (2^r-7)}\\
			&	&&\hspace{5em} + 2\cdot \set{(2^r-4)(2^r-8)+ (2^r-4)} ]\\
			&	&&= (2^r-1)(2^r-2)(2^r-4)(2^r-7).
		\end{aligned}
	\end{equation*}

	Now, we will consider the scenario where the equation $x^3 = 1$ has no solutions~\footnote{For example, over the finite field $\mathbb{F}_{2^5}$ other than $x = 1$, the equation $x^3 = 1$ has no solutions other than $x = 1$.} in $\mathbb{F}_{2^r}$ other than the $x=1$. In this scenario, we only consider the case $b \neq a\cdot x$. Thus, we have

	\begin{equation*}
		\begin{aligned}
			|S|&= (2^r-1)(2^r-2)[4\cdot (2^r-7)+(2^r-7)(2^r-8)]\\
			&= (2^r-1)(2^r-2)(2^r-4)(2^r-7).
		\end{aligned}
	\end{equation*}

	Therefore, the number of $4\times 4$ Hadamard MDS matrices over the finite field $\mathbb{F}_{2^r}$ is equal to $(2^r-1)(2^r-2)(2^r-4)(2^r-7)$. \qed
\end{proof}


\noindent It is worth noting that the Hadamard matrix $\hada(a,b,c,d)$ over $\mathbb{F}_{2^r}$ achieves involutory property if and only if the condition $a+b+c+d=1$ is satisfied. In the following theorem, we provide the exact count of the $4\times 4$ involutory Hadamard MDS matrices over $\mathbb{F}_{2^r}$.





\begin{theorem}\label{Thm_Inv_Hadamard_MDS_Counting}
    For $r\geq 3$, the count of $4\times 4$ involutory Hadamard MDS matrices over the finite field $\mathbb{F}_{2^r}$ is given by \[(2^r-2)(2^r-4)(2^r-7).\]
\end{theorem}

\begin{proof}
    If $M = \hada(a, b, c, d)$ is a Hadamard matrix, then any nonzero scalar multiple of $M$ is also a Hadamard matrix. Now since $M$ is MDS, we must have $a + b + c + d \neq 0$. Thus, for every $t \neq 0$, there exist an equal number of Hadamard matrices such that $a + b + c + d = t$, as one can freely scale them, resulting in a bijection. In particular, the number of Hadamard matrices with a row sum of 1 is exactly the total number of Hadamard matrices divided by $(2^r - 1)$. Therefore, according to Theorem~\ref{Th_Hadamard_MDS_Counting}, the count of $4 \times 4$ involutory Hadamard MDS matrices is given by $(2^r - 2)(2^r - 4)(2^r - 7)$. \qed
\end{proof}

From Theorems~\ref{Th_Hadamard_MDS_Counting} and \ref{Thm_Inv_Hadamard_MDS_Counting}, one can easily find the number of non-involutory Hadamard MDS matrices over the finite field $\mathbb{F}_{2^r}$. We present the count in the following corollary.

\begin{corollary}
	For $r\geq 3$, the count of $4\times 4$ non-involutory Hadamard MDS matrices over the finite field $\mathbb{F}_{2^r}$ is given by $(2^r-2)^2(2^r-4)(2^r-7)$.
\end{corollary}

In Table~\ref{Table_Hadamard_count}, we provide the exact counts for $4\times 4$ Hadamard MDS matrices, involutory Hadamard MDS matrices, and non-involutory Hadamard MDS matrices over $\mathbb{F}_{2^r}$ for $r=3, 4, \ldots, 8$. Appendix~\ref{Sec:Appendix} provides the full list of $4\times 4$ Hadamard MDS and involutory Hadamard MDS matrices over $\mathbb{F}_{2^3}$.

\begin{table}[h!]
	\centering
	\begin{tabular}{|c|c|c|c|}
	\hline
	$\mathbb{F}_{2^r}$ & Hadamard MDS & involutory Hadamard MDS & non-involutory Hadamard MDS \\
	\hline
	$\mathbb{F}_{2^3}$ & 168 & 24 & 144 \\
	$\mathbb{F}_{2^4}$ & 22680 & 1512 & 21168 \\
	$\mathbb{F}_{2^5}$ & 651000 & 21000 & 630000 \\
	$\mathbb{F}_{2^6}$ & 13358520 & 212040 & 13146480 \\
	$\mathbb{F}_{2^7}$ & 240094008 & 1890504 & 238203504 \\
	$\mathbb{F}_{2^8}$ & 4064187960 & 15937992 & 4048249968 \\
	\hline
	\end{tabular}
	\vspace{1em}
	\caption{The count of $4\times 4$ Hadamard MDS, involutory Hadamard MDS and non-involutory Hadamard MDS matrices.\label{Table_Hadamard_count}}
\end{table}

\noindent It can be observed that if any two elements in the first row of a Hadamard matrix $\hada(a_1, \ldots, a_n) $ are equal, then the matrix contains a singular $2 \times 2$ submatrix. Consequently, it cannot be an MDS matrix. We mention this result as follows.

\begin{proposition}
	In a Hadamard MDS matrix $\hada(a_1, \ldots, a_n) $, it holds that $a_i \neq a_j $ for $i, j = 1, \ldots, n $.
\end{proposition}


Thus, in a Hadamard MDS matrix of order $n$, the number of distinct elements is exactly equal to $n$. In contrast, a circulant MDS matrix can have repeated entries in its first row. For instance, in the AES~\cite{AES} diffusion matrix, which is a circulant MDS matrix, the first row contains two entries of $1$. This repetition makes such matrices particularly valuable in lightweight cryptography.

Next, similar to $4 \times 4$ Hadamard MDS matrices, we aim to derive an explicit formula for the count of $4 \times 4$ circulant MDS matrices over the finite field $\mathbb{F}_{2^r}$. However, we are currently unable to formulate this due to constraints on the variables. Specifically, the $4 \times 4$ circulant matrix $\text{Circ}(a,b,c,d)$ will be an MDS matrix if and only if each element of the set
\begin{center}
	\begin{small}
		$\{ \splitatcommas{a, b, c, d, a^2b + bc^2 + b^2d + d^3, ab + cd, a + c, a^3 + b^2c + ac^2 + cd^2, ab^2 + a^2c + c^3 + ad^2, b^2 + ac, b + d, b^3 + a^2d + c^2d + bd^2, c^2 + bd, bc + ad, a^2 + bd, ac + d^2, a+b+c+d} \}$
	\end{small}
\end{center}
is nonzero. Consequently, we have provided experimental results comparing $4 \times 4$ Hadamard and circulant MDS matrices over $\mathbb{F}_{2^r}$. Table~\ref{Table_comparison_count} presents a comparison of their counts.


\begin{table}[h!]
	\centering
	\begin{tabular}{|c|c|c|c|c|}
	\hline
	$\mathbb{F}_{2^r}$ & Hadamard MDS & involutory Hadamard MDS & circulant MDS\\
	\hline
	$\mathbb{F}_{2^3}$ & 168 & 24 & 0 \\
	$\mathbb{F}_{2^4}$ & 22680 & 1512 & 16560 \\
	$\mathbb{F}_{2^5}$ & 651000 & 21000 & 580320 \\
	$\mathbb{F}_{2^6}$ & 13358520 & 212040 & 12685680 \\
	$\mathbb{F}_{2^7}$ & 240094008 & 1890504 & 234269280 \\
	$\mathbb{F}_{2^8}$ & 4064187960 & 15937992 & 4015735920 \\
	\hline
	\end{tabular}
	\vspace{1em}
	\caption{Comparison of the count of $4\times 4$ Hadamard MDS and circulant MDS matrices\label{Table_comparison_count}}
\end{table}

\section{$2n\times 2n$ involutory MDS matrices}\label{Sec:general_2n_MDS}
In this section, we first present a general form for all $2n\times 2n$ involutory MDS matrices over $\mathbb{F}_{2^r}$. Then, we determine the exact count of $2\times 2$ MDS and involutory MDS matrices over $\mathbb{F}_{2^r}$. Using these counts, we establish an upper bound for the enumeration of all $4\times 4$ involutory MDS matrices over $\mathbb{F}_{2^r}$.


\begin{lemma}\label{Lemma_structure_Inv_MDS}
	Let $M=
	\begin{bmatrix}
		A_1 & A_2\\
		A_3 & A_4
	\end{bmatrix}$
	be a $2n\times 2n$ involutory MDS matrix over $\mathbb{F}_{2^r}$, where $A_i$ are $n\times n$ matrices. Then $M$ can be expressed in the form $M=DHD^{-1}$, where $H=
	\begin{bmatrix}
		A_1 & I_n+A_1\\
		I_n+A_1 & A_1
	\end{bmatrix}$
	and $D$ is the block diagonal matrix $\diag(I_n,A_3(I_n+A_1)^{-1})$.
\end{lemma}

\begin{proof}
	Since $M$ is involutory we have 
	\begin{align}
		& A_1^2+A_2A_3=I_n,  && A_1A_2+A_2A_4= \mathbf{0}, \notag \\
		& A_3A_1+A_4A_3=\mathbf{0}, && A_3A_2+A_4^2=I_n. \notag \\
		\implies & A_2=(I_n+A_1^2)A_3^{-1} && A_4= A_3A_1A_3^{-1}. \notag
	\end{align}
	
    \noindent Now, as $M$ is an MDS matrix, $A_2$ is nonsingular. Thus, we must have $I_n+A_1^2$ nonsingular. Further, since $I_n+A_1^2=(I_n+A_1)^2$, it follows that $I_n+A_1$ is nonsingular as well. Also, $A_1 \neq I_n$. Therefore, $M$ can be expressed as:
	\begin{equation*}
		\begin{aligned}
			M & =
			\begin{bmatrix}
				A_1 & (I_n+A_1^2)A_3^{-1}\\
				A_3 & A_3A_1A_3^{-1}
			\end{bmatrix}\\
			& =
			\begin{bmatrix}
				I_n & \mathbf{0}\\
				\mathbf{0} & A_3(I_n+A_1)^{-1}
			\end{bmatrix}
			\begin{bmatrix}
				A_1 & I_n+A_1\\
				I_n+A_1 & A_1
			\end{bmatrix}
			\begin{bmatrix}
				I_n & \mathbf{0}\\
				\mathbf{0} & (I_n+A_1)A_3^{-1}
			\end{bmatrix}\\
			& = DHD^{-1},
		\end{aligned}
	\end{equation*}
	where $H$ is the block matrix 
	$\begin{bmatrix}
		A_1 & I_n+A_1\\
		I_n+A_1 & A_1
	\end{bmatrix}$
	and $D$ is the block diagonal matrix $\diag(I_n,A_3(I_n+A_1)^{-1})$. \qed
\end{proof}

\begin{lemma}\label{Lemma_2_MDS_Counting}
	The count of $2\times 2$ MDS matrices over the finite field $\mathbb{F}_{2^r}$ is given by $(2^r-1)^3(2^r-2)$.
\end{lemma}

\begin{proof}
	Let 
	$M= \begin{bmatrix}
		a & b \\
		c & d 
	\end{bmatrix}$ 
	be an MDS matrix over $\mathbb{F}_{2^r}$.
	Note that $a$, $b$, $c$, and $d$ all belong to the nonzero elements of the finite field $\mathbb{F}_{2^r}$. Also, $\det (M)=ad+bc$ must be nonzero. This implies that $d\neq a^{-1}bc$. Consequently, each of the values $a$, $b$, and $c$ can be chosen from $2^r - 1$ possibilities, as they are drawn from the nonzero elements of $\mathbb{F}_{2^r}$. On the other hand, $d$ can be selected from $2^r - 2$ possibilities because it cannot be equal to $a^{-1}bc$.
	Hence, the total number of $2 \times 2$ MDS matrices can be calculated as $(2^r - 1)^3(2^r - 2)$. \qed
\end{proof}

\begin{lemma}\label{Lemma_Inv_2_MDS_Counting}
	The count of $2\times 2$ involutory MDS matrices over the finite field $\mathbb{F}_{2^r}$ is given by $(2^r-1)(2^r-2)$.
\end{lemma}

\begin{proof}
	According to Lemma~\ref{Lemma_structure_Inv_MDS}, any $2\times 2$ involutory MDS matrix $M$ can be represented as $M=DHD^{-1}$, where $H$ is the $2\times 2$ Hadamard matrix $\hada(\alpha,1+\alpha)$, and $D$ is a nonsingular diagonal matrix $\diag(1,(1+\alpha)\beta^{-1})$, with $\alpha,\beta \in \mathbb{F}_{2^r}^\star$. Also, from Corollary~\ref{Corollary_DMD_MDS}, we can say that $M$ is an involutory MDS if and only if $H$ is an involutory MDS.

	Regarding $H$, there are $2^r-2$ possible choices since $\alpha \not \in \set{0,1}$. On the other hand, $D$ provides $2^r-1$ options. Therefore, the total number of $2\times 2$ involutory MDS matrices is $(2^r-1)(2^r-2)$. \qed
\end{proof}

The count of $2\times 2$ involutory MDS matrices over $\mathbb{F}_{2^r}$ is also provided in \cite{Sakalli2020}. However, it is noteworthy that the enumeration of $2\times 2$ involutory MDS matrices in this paper is based on the generic matrix form outlined in Lemma~\ref{Lemma_structure_Inv_MDS}, whereas the count in \cite{Sakalli2020} relies on the characterization of $2\times 2$ involutory MDS matrices.\\

\noindent In Theorem~\ref{Thm_Inv_Hadamard_MDS_Counting}, we provide the count of $4\times 4$ involutory Hadamard MDS matrices. Next, we establish an upper bound for the enumeration of all $4\times 4$ involutory MDS matrices.

\begin{proposition}\label{Prop_4_MDS_Counting}
	For $r\geq 3$, the number of $4\times 4$ involutory MDS matrices over the finite field $\mathbb{F}_{2^r}$ is upper bounded by $2^r(2^r-1)^3(2^r-2)^2(2^r-3)(2^r-4)$.
\end{proposition}

\begin{proof}
	From Lemma~\ref{Lemma_structure_Inv_MDS}, we can deduce that any $4\times 4$ involutory MDS matrix $M$ can be expressed as: 
	\begin{equation*}
		\begin{aligned}
			M=
			\begin{bmatrix}
				A_1 & (I_n+A_1^2)A_3^{-1}\\
				A_3 & A_3A_1A_3^{-1}
			\end{bmatrix},
		\end{aligned}
	\end{equation*}
	where both $A_1$ and $A_3$ are $2\times 2$ MDS matrices. Additionally, $A_1$ cannot be an involutory matrix. Therefore, by considering both Lemma~\ref{Lemma_2_MDS_Counting} and Lemma~\ref{Lemma_Inv_2_MDS_Counting}, we can conclude that $A_1$ has a total of $(2^r-1)^3(2^r-2)-(2^r-1)(2^r-2)=2^r(2^r-1)(2^r-2)^2$ possible choices.

	Each row of $A_3$ must be linearly independent with the rows of $A_1$. Since $A_3$ is a $2\times 2$ MDS matrix, there are $(2^r-1)(2^r-3)$ possible choices for the first row of $A_3$ and $(2^r-1)(2^r-4)$ possible choices for the second row. Thus, $A_3$ has a total of $(2^r-1)^2(2^r-3)(2^r-4)$ possible choices. Hence, the number of $4\times 4$ involutory MDS matrices over the finite field $\mathbb{F}_{2^r}$ is upper bounded by $2^r(2^r-1)^3(2^r-2)^2(2^r-3)(2^r-4)$. \qed
\end{proof}

It is noteworthy that there has been recent work on the enumeration of $4 \times 4$ involutory MDS matrices. In \cite{Kumar_4MDS} and \cite{Tuncay23}, the authors introduced hybrid methods for studying these matrices over $\mathbb{F}_{2^r}$. Importantly, in the hybrid approach, a representative involutory MDS matrix is typically derived through a search process, which subsequently serves as a basis for generating multiple involutory MDS matrices. Specifically, \cite{Tuncay23} significantly reduced the search space for finding all $4 \times 4$ involutory MDS matrices by identifying the representative involutory MDS matrices of order 4, thereby reducing the search space from $2^{16r}$ to $2^{8r}$. Following this, \cite{Kumar_4MDS} further refined the search space to $2^{5r}$. By conducting an exhaustive search within this optimized search space, the authors provided the counts of $4 \times 4$ involutory MDS matrices over $\mathbb{F}_{2^r}$ for $r = 3, 4, \ldots, 8$. Note that, as derived in \cite{Kumar_4MDS} and \cite{Tuncay23}, the count for $4\times 4$ involutory MDS matrices over $\mathbb{F}_{2^3}$ and $\mathbb{F}_{2^4}$ are 16,464 and 242,514,000, respectively.

\begin{remark}
    As discussed above, it is crucial to emphasize that the upper bound derived in Proposition~\ref{Prop_4_MDS_Counting} is very loose. For instance, while the count of all $4\times 4$ involutory MDS matrices over $\mathbb{F}_{2^3}$ is 16,464, our bound yields $19,75,680$. This is because our calculation only considers cases involving $2\times 2$ submatrices constructed from $A_1$ and $A_3$. To achieve a more precise bound, it is necessary to examine the other $2\times 2$ submatrices as well as all $3\times 3$ submatrices. However, currently, we are unable to exclude these potential cases from our analysis.
\end{remark}

\begin{remark}
    As stated in Lemma~\ref{Lemma_structure_Inv_MDS}, we can represent any $4\times 4$ involutory MDS matrix $M$ as $M=DHD^{-1}$, where $H$ is a $4\times 4$ involutory matrix. However, it is important to note that in this context, $D$ is not a standard diagonal matrix but rather a diagonal block matrix. Consequently, unlike Lemma~\ref{Lemma_Inv_2_MDS_Counting} for $2\times 2$ involutory MDS matrices, we cannot directly apply Corollary~\ref{Corollary_DMD_MDS} to assert that $M$ is MDS if and only if $H$ is MDS. Furthermore, considering the specific form of $H$, it cannot be considered an MDS matrix.
\end{remark}

\section{Some results on Near-MDS matrices}\label{Sec:Near-MDS}

In this section, we discuss some results related to Near-MDS (NMDS) matrices. In~\cite{Combo_NMDS_2023}, the authors investigated various matrix forms, including Hadamard, circulant, and left-circulant matrices, for the construction of NMDS matrices. However, there are some circulant-like matrices that remain to be studied in relation to NMDS matrices. Therefore, in this section, we study these matrices for the construction of NMDS matrices.

Furthermore, we know that an MDS matrix is always nonsingular, while the authors in~\cite{Combo_NMDS_2023} demonstrated that an NMDS matrix may be singular. Below, we will show that if a Hadamard matrix is singular, it cannot be an NMDS matrix.

\begin{theorem}\label{Th_Hadamard_singular}
	An NMDS Hadamard matrix over $\mathbb{F}_{2^r}$ is always nonsingular.
\end{theorem}
\begin{proof}
	Let $M=\hada(a_1,a_2,\ldots,a_n)$ be a Hadamard matrix of order $n$ over $\mathbb{F}_{2^r}$. We know that $M^2 = c^2 I_n$, where $c = a_1 + a_2 + \cdots + a_n$, and $\det(M) = 0$ if and only if $c = 0$.
	
	Now, we know that, by Cayley-Hamilton theorem, the adj$(M)$ (the adjugate matrix of $M$) can be expressed in terms of powers of $M$ as
	\begin{equation*}
		\begin{aligned}
			\text{adj}(M)&= M^{n-1}+c_{n-1} M^{n-2} + \cdots + c_1 M + c_1I_n
		\end{aligned}
	\end{equation*}
	where $t^n + c_{n-1} t^{n-1} + \cdots + c_1 t + c_0$ is the characteristic polynomial of the matrix $M$.

	However, since $M^2=c^2I_n$, the eigenvalues of $M$ are $c$. Thus, the characteristic polynomial of the matrix $M$ is given by $(x+c)^n$. Since $M$ is a Hadamard matrix, $n$ is a power of $2$ which implies that the characteristic polynomial of $M$ is given by $x^n+c^n$.

	Therefore, we have
	\begin{equation*}
		\begin{aligned}
			\text{adj}(M)&= M^{n-1}= M^{n-2}\cdot M\\
			&= c^{n-2}I_n\cdot M \quad [\text{ since }M^2=c^2I_n]\\
			&= c^{n-2}M.
		\end{aligned}
	\end{equation*}

	Thus, each element in $\text{adj}(M)$ has a factor term of $c^{n-2}$. We also know that each element in $\text{adj}(M)$ corresponds to the determinant of the $(n-1) \times (n-1)$ submatrices of $M$. Therefore, we can conclude that the determinant of each $(n-1) \times (n-1)$ submatrix of $M$ is of the form $\alpha \cdot c^{n-2}$, for some $\alpha \in \mathbb{F}_{2^r}$. Consequently, if $M$ is singular, then each $(n-1) \times (n-1)$ submatrix of $M$ must also be singular. Thus, from Lemma~\ref{Lemma_nMDS_charac}, we conclude that $M$ can never be an NMDS matrix.
	\qed
\end{proof}





The absence of zero entries is a necessary condition for matrices such as Hadamard and circulant matrices to be MDS. Consequently, these matrices incur a high implementation cost. In contrast, having zero entries (with a maximum of one zero per row or column) does not affect the NMDS property of these matrices, resulting in a reduced implementation cost. More specifically, in Lemma~\ref{Lemma_nMDS_charac}, if we assume $g=1$, we can deduce that there is at most one zero in each row and each column of an NMDS matrix.


Leveraging this advantage, the authors in~\cite{Li_Wang_2017} introduced several generic lightweight involutory NMDS matrices of order $8$ derived from Hadamard matrices. Also, in~\cite{Combo_NMDS_2023}, the authors discussed the minimum implementation cost of Hadamard and circulant NMDS matrices. In the subsequent discussion, we explore the results related to $4 \times 4$ Hadamard and circulant NMDS matrices with one zero per row or column.

\begin{lemma}\label{Lemma_4Hadamard_NMDS}
	Any nonsingular $4 \times 4$ Hadamard matrix $\hada(0, a_1, a_2, a_3)$, where $a_i \in \mathbb{F}_{2^r}^*$ for $i = 1, 2, 3$, is an NMDS matrix over the field $\mathbb{F}_{2^r}$.
\end{lemma}

\begin{proof}
	Let $M=\hada(0,a_1,a_2,a_3)$ be a nonsingular matrix. Then we have $a_1+a_2+a_3\neq 0$. Since $a_i\neq 0$ and $H$ is symmetric, to show $M$ is an NMDS matrix, we need to show that each $2\times 3$ and $3\times 4$ submatrices is of full rank.

	First, consider any $2\times 3$ submatrix. It contains a $2\times 2$ submatrix which contain either two zero (situated diagonally or anti-diagonally) or one zero. Since $a_i\neq 0$, this $2\times 2$ submatrix is always nonsingular. Therefore, each $2\times 3$ submatrix is of full rank.

	Now, as discussed in the proof of Theorem~\ref{Th_Hadamard_singular}, the determinant of any $3\times 3$ submatrices is of the form 
	\[(a_1+a_2+a_3)^2\cdot (M)_{i,j}.\]
	Since in any row (or column), $(M)_{i,j}= 0$ exactly once, any $3\times 4$ submatrix has three nonsingular $3\times 3$ submatrix i.e., in other words any $3\times 4$ submatrix is of full rank. Hence, by Lemma~\ref{Lemma_nMDS_charac}, $M$ is an NMDS matrix.
	\qed
\end{proof}

In the above lemma, the position of zeros can be adjusted such that there is exactly one zero per row and per column. This can be achieved through multiplying some permutation matrix to the matrix $M$. Thus, by Lemma~\ref{Lemma_PMQ_MDS}, $M$ will remain an NMDS matrix. Therefore, we can conclude that in a $4 \times 4$ Hadamard matrix $\hada(a_1, a_2, a_3, a_4)$, if $a_1 + a_2 + a_3 + a_4 \neq 0$ and exactly one $a_i = 0$, the matrix will remain an NMDS matrix. We state this result formally in the following proposition.

\begin{proposition}
	Any nonsingular $4 \times 4$ Hadamard matrix $\hada(a_1, a_2, a_3, a_4)$, where exactly one $a_i = 0$, is always an NMDS matrix.
\end{proposition}

\begin{corollary}\label{Coro_4Hadamard_counting}
	The number of $4 \times 4$ Hadamard NMDS matrices $\hada(\splitatcommas{a_1, a_2, a_3, a_4})$, with exactly one $a_i = 0$, over $\mathbb{F}_{2^r}$ is given by 
	\[4(2^r - 1)(2^{2r}-3\cdot 2^r+3).\]
\end{corollary}

\begin{proof}
	We can choose the position of zero in the Hadamard matrix in four different ways. Let us assume that $a_1 = 0$.  Now, for the Hadamard matrix $\hada(0, a_2, a_3, a_4)$ to be NMDS, it must be nonsingular. Thus, we require
	\[
		a_2 + a_3 + a_4 \neq 0.
	\]

	Since $a_2, a_3, a_4 \in \mathbb{F}_{2^r}^*$, each of $a_2$ and $a_3$ has $2^r - 1$ choices over $\mathbb{F}_{2^r}$. Additionally, for the matrix to remain nonsingular, we must have $a_4 \neq a_2 + a_3$. 

	However, when $a_2 = a_3$, we find $a_2 + a_3 = 0$. In this case, we can choose $a_4$ in $2^r - 1$ ways. Thus, we have two cases: when $a_3 \neq a_2$, $a_4$ has $2^r - 2$ choices, and when $a_3 = a_2$, $a_4$ has $2^r - 1$ choices.

	Therefore, for the Hadamard matrix to be NMDS, the total number of possible choices for $a_2, a_3, a_4$ is given by

	\begin{equation*}
		\begin{aligned}
			&(2^r - 1) \left[ (2^r - 2) \cdot (2^r - 2) + 1 \cdot (2^r - 1) \right] \\
			&= (2^r - 1) \left( 2^{2r} - 3 \cdot 2^r + 3 \right).
		\end{aligned}
	\end{equation*}

	Hence, the number of $4 \times 4$ Hadamard NMDS matrices $\hada(a_1, a_2, a_3, a_4)$, with exactly one $a_i = 0$, over $\mathbb{F}_{2^r}$ is given by 
	\[
	4(2^r - 1)(2^{2r} - 3 \cdot 2^r + 3).
	\]
	\qed
\end{proof}

Similar to the proof of Theorem~\ref{Thm_Inv_Hadamard_MDS_Counting}, we can assert that the number of Hadamard NMDS matrices with a row sum of 1 is exactly the total number of Hadamard NMDS matrices divided by $2^r - 1$. Therefore, by the above corollary, we conclude that the number of $4 \times 4$ involutory Hadamard NMDS matrices $\hada(a_1, a_2, a_3, a_4)$ with exactly one $a_i = 0$ over $\mathbb{F}_{2^r}$ is $4(2^{2r}-3\cdot 2^r+3)$.

\begin{corollary}\label{Coro_4INVHadamard_counting}
	The number of $4 \times 4$ involutory Hadamard NMDS matrices $\hada(\splitatcommas{a_1, a_2, a_3, a_4})$, with exactly one $a_i = 0$, over $\mathbb{F}_{2^r}$ is given by
	\[4(2^{2r}-3\cdot 2^r+3).\]
\end{corollary}

Now, we will discuss the result related to $4\times 4$ circulant NMDS matrices with one zero per row or column.

\begin{lemma}\label{Lemma_4circulant_NMDS}
	Any $4\times 4$ circulant matrix $Circ(0,a_1,a_2,a_3)$, where $a_i \in \mathbb{F}_{2^r}^*$ for $i = 1, 2, 3$, is an NMDS matrix over the field $\mathbb{F}_{2^r}$.
\end{lemma}

\begin{proof}
	Let $M = \text{Circ}(0, a_1, a_2, a_3)$ be a $4 \times 4$ circulant matrix over $\mathbb{F}_{2^r}$. Since each $a_i \in \mathbb{F}_{2^r}^*$, we can assert that any $1 \times 2$ (or $2 \times 1$) submatrix is of full rank. Thus, we need to show that each $2 \times 3$ (or $3 \times 2$) and $3 \times 4$ (or $4 \times 3$) submatrix is of full rank.

	Due to the structure of $M$ any $2 \times 3$ (or $3 \times 2$) contain a $2\times 2$ submatrix which contain either two zero (situated diagonally or anti-diagonally) or one zero. Since $a_i\neq 0$, this $2\times 2$ submatrix is always nonsingular. Therefore, each $2\times 3$  (or $3 \times 2$) submatrix is of full rank.

	Consider the $3\times 4$ submatrix
	\begin{equation*}
		\begin{aligned}
			N=&
			\begin{bmatrix}
				0 & a_{1} & a_{2} & a_{3} \\
				a_{3} & 0 & a_{1} & a_{2} \\
				a_{2} & a_{3} & 0 & a_{1}
			\end{bmatrix}.
		\end{aligned}
	\end{equation*}

	Now, if we consider the $3 \times 3$ submatrix $$N' = \begin{bmatrix} 0 & a_2 & a_3 \\ a_3 & a_1 & a_2 \\ a_2 & 0 & a_1 \end{bmatrix}$$ obtained from $N$ by deleting the second column of $N$, we have $\det(N') = a_2^3 \neq 0$. Therefore, $N$ is of full rank.


	Also, all the other three $3 \times 4$ submatrices of $M$ are of the structure obtained by column permutation on the submatrices of structure $N$. Therefore, each of them has rank $3$. Thus, any $3 \times 4$ submatrix of $M$ is of full rank. Similarly, we can show that any $4 \times 3$ submatrix of $M$ is of full rank. Therefore, by Lemma~\ref{Lemma_nMDS_charac}, we conclude that $M$ is an NMDS matrix.
	\qed
\end{proof}

Like Hadamard matrices, the zeros in the above lemma can be arranged so that there is exactly one zero in each row and each column. Thus, we have the following result.


\begin{proposition}
	Any $4 \times 4$ circulant matrix $M = Circ(a_1, a_2, a_3, a_4)$, where exactly one $a_i = 0$, is always an NMDS matrix.
\end{proposition}

\begin{corollary}\label{Coro_4circulant_counting}
	The number of $4 \times 4$ circulant NMDS matrices $Circ(a_1, a_2, a_3, a_4)$, with exactly one $a_i = 0$, over $\mathbb{F}_{2^r}$ is given by $4(2^r-1)^3$.
\end{corollary}

\begin{remark}\label{Remark_4NMDS}
	From Lemma~\ref{Lemma_4Hadamard_NMDS}, we know that for any values of $a_1, a_2, a_3$ satisfying $a_1 + a_2 + a_3 \neq 0$, we can obtain an NMDS matrix. By setting $a_1 = a_2 = a_3 = 1$, we recover the NMDS matrices used in~\cite{MIDORI,SKINNY}. This is also the circulant matrix derived from Lemma~\ref{Lemma_4circulant_NMDS}.
\end{remark}

\begin{remark}
	From Theorem~\ref{Th_Hadamard_singular}, we know that for a Hadamard matrix to be NMDS, it must be nonsingular. In contrast, circulant NMDS matrices can be singular. For instance, over $\mathbb{F}_{2^4}$, there are $840$ singular circulant NMDS matrices of order $4$ with one zero per row or column.
\end{remark}

\noindent It is noteworthy that the circulant MDS matrix $Circ(\alpha, 1+\alpha, 1, 1)$, which is utilized in the AES MixColumn operation, consists of elements with low multiplication cost; however, this matrix contains a total of $8$ entries that are equal to $1$. In~\cite{PASCAL}, Junod et al. established that the highest number of $1$'s in a $4 \times 4$ MDS matrix is $9$. To reach this result, they introduced a new class of efficient MDS matrices that incorporate circulant submatrices. Additionally, Gupta et al. formalized this as Type-I circulant-like matrices in~\cite{GR15}, conducting a comprehensive examination of these matrices for the development of efficient and perfect diffusion layers. Below, we present the definition of Type-I circulant-like matrices as outlined in~\cite{GR15,PASCAL}.

\begin{definition}\cite{GR15,PASCAL}\label{clm1}
	The $n \times n$ matrix
	$$
	\begin{bmatrix}
		a & \mathbf{1} \\
		\textbf{1}^T & A
	\end{bmatrix}
	$$
	is referred to as a Type-I circulant-like matrix, where $A = Circ(1, a_1, \ldots, a_{n-2})$, $\mathbf{1} = \underbrace{(\splitatcommas{1, \ldots, 1})}_{\text{\tiny{n-1 times}}}$, $1$ represents the unit element, and both $a_i$'s and $a$ are any nonzero elements from the underlying field other than $1$. This matrix is denoted as $TypeI(a, Circ(1, a_1, \ldots, a_{n-2}))$.
\end{definition}

In~\cite{GR15}, it is established that there is no involutory Type-I circulant-like matrix of even order, while~\cite{MDS_Survey} demonstrated the absence of such matrices for odd order as well. We combine these findings in the following lemma.

\begin{lemma}\cite{MDS_Survey,GR15}
	Any Type-I circulant-like matrix over $\mathbb{F}_{2^r}$ cannot be involutory. 
\end{lemma}

In~\cite{MDS_Survey}, it was shown that a Type-I circulant-like matrix of odd order cannot be orthogonal, while~\cite{GR15} demonstrated that a Type-I circulant-like matrix of even order cannot be both orthogonal and MDS. However, we will now show that this result does not hold for NMDS matrices. Specifically, in the following result, we prove that there are exactly two $4 \times 4$ orthogonal Type-I circulant-like NMDS matrices over the field $\mathbb{F}_{2^r}$.

\begin{proposition}
    There exist exactly two $4 \times 4$ orthogonal Type-I circulant-like NMDS matrices over $\mathbb{F}_{2^r}$.
\end{proposition}

\begin{proof}
	Let $M= TypeI(a, Circ(1, a_1,a_2))$ be an orthogonal matrix over $\mathbb{F}_{2^r}$. Thus, $M\cdot M^T=I_4$ implies that $a^2+1=1\implies a=0$.

	Also, given $a = 0$, we have the following system of equations:
	\begin{equation*}
		\begin{aligned}
			& a_1^2+a_2^2=1\\
			& a_1a_2+a_1+a_2+1=0\\
			\implies & a_1+a_2=1 \quad \text{ and } \quad (a_1+1)(a_2+1)=0\\
			\implies & (a_1,a_2)=(0,1) \quad \text{ or }\quad (a_1,a_2)=(1,0).
		\end{aligned}
	\end{equation*}
	Thus, $M$ takes the form
	\[
	\begin{bmatrix}
		0 & 1 & 1 & 1 \\
		1 & 1 & 0 & 1 \\
		1 & 1 & 1 & 0 \\
		1 & 0 & 1 & 1
	\end{bmatrix} ~~~\text{or}~~
	\begin{bmatrix}
		0 & 1 & 1 & 1 \\
		1 & 1 & 1 & 0 \\
		1 & 0 & 1 & 1 \\
		1 & 1 & 0 & 1
	\end{bmatrix}.
	\]

	Also, these two matrices are row-permuted versions of the NMDS matrices discussed in Remark~\ref{Remark_4NMDS}, and hence, they are NMDS matrices. Therefore, there are exactly two $4 \times 4$ orthogonal Type-I circulant-like NMDS matrices over $\mathbb{F}_{2^r}$.
	\qed
\end{proof}

\noindent Involutory MDS (or NMDS) matrices are of particular interest. However, circulant matrices cannot simultaneously be both involutory and MDS. To address this limitation, the authors in~\cite{GR15} proposed a new category of circulant-like matrices that are inherently involutory, referring to them as Type-II circulant-like matrices.

\begin{definition}\label{clm2}\cite{GR15}
	The $2n \times 2n$ matrix
	$$
	\begin{bmatrix}
	A & A^{-1} \\
	A^3 + A & A \\
	\end{bmatrix}
	$$
	is called a Type-II circulant-like matrix, where $A = Circ(a_0, \ldots, a_{n-1})$. This matrix is denoted as $TypeII(Circ(a_0, \ldots, a_{n-1}))$.
\end{definition}

The authors of \cite{GR15} demonstrated that when $n$ is even, the corresponding $2n \times 2n$ Type-II circulant-like matrix is not MDS. In the following theorem, we will show that this result also holds for NMDS matrices.

\begin{theorem}\label{Th_TypeII_NMDS}
	Any $2n \times 2n$ Type-II circulant-like matrix over $\mathbb{F}_{2^r}$ can never be an NMDS matrix for even values of $n$.
\end{theorem}

\begin{proof}
	Let $M = TypeII(Circ(a_0, \ldots, a_{n-1})$) be a $2n \times 2n$ Type-II circulant-like matrix, where $n = 2d$. Assume that $A = Circ(a_0, \ldots, a_{2d-1})$. Therefore, we have
	\begin{equation*}
		\begin{aligned}
			A^2 
			&= Circ(a_0^2 + a_d^2, 0, a_1^2 + a_{d+1}^2, 0, \ldots, a_{d-1}^2 + a_{2d-1}^2, 0).\\
			&= Circ(b_0, 0, b_1, 0, \ldots, b_{d-1}, 0),
		\end{aligned}
	\end{equation*}
	where $b_i = a_i^2 + a_{d+i}^2$ for $i = 0, \ldots, d-1$. Therefore, the first column of $A^2$ is $[b_0,0,b_{d-1},\ldots,b_{1},0]^T$.
	
	Now, let $A^3= Circ(e_0, e_1, \ldots, e_{2d-1})$ and denote the columns of $A$ as $C_0, C_1, \ldots, C_{2d-1}$. Now, we will show that the first column of $A^3$ are the linear combination of $C_0,C_2,\ldots,C_{2d-2}$.

	Since $A^3= A\cdot A^2$, the first column of $A^3$ is given by
	\begin{equation*}
		\begin{aligned}
			& \begin{bmatrix}
				e_0\\
				e_{2d-1}\\
				e_{2d-2}\\
				\vdots\\
				e_2\\
				e_1\\
			\end{bmatrix}
			= A \cdot 
			\begin{bmatrix}
				b_0\\
				0\\
				b_{d-1}\\
				\vdots \\
				b_1\\
				0\\
			\end{bmatrix}\\
			&= b_0 \cdot C_0 + b_{d-1} \cdot C_2+ \cdots + b_1\cdot C_{2d-2}.
		\end{aligned}
	\end{equation*}

	Therefore, the first column of $A^3 + A$ is given by
	$$
	\mathcal{C}_0 = (b_0 + 1)C_0 + b_{d-1}C_2 + b_{d-2}C_4 + \ldots + b_1 C_{2d-2}.
	$$

	This expression shows that $\mathcal{C}_0$ is a linear combination of the columns $C_0, C_2, \ldots, C_{2d-2}$. Consequently, the set of $d+1$ columns, $\mathcal{C}_0$ along with $C_0, C_2, \ldots, C_{2d-2}$, are linearly dependent.

	Now consider the $(d+2) \times (d+1)$ submatrix $N$ of $M$, constructed from the rows $R_{2d}$, $R_{2d+1},\ldots, R_{2d+d+1}$ and the columns $T_0$, $T_{2d}$, $T_{2d+2},\ldots, T_{2d+2d-2}$, where $R_i$ denotes the $i$-th row and $T_j$ denotes the $j$-th column of $M$.

	As shown earlier, the $d+1$ columns of $N$ are linearly dependent. Moreover, any $(d+1) \times (d+1)$ submatrix of $N$ consists of the same $d+1$ columns of $N$, but with one row omitted, and thus will also be linearly dependent. In other words, the $(d+2) \times (d+1)$ submatrix $N$ has rank less than $d+1$. Therefore, by Lemma~\ref{Lemma_nMDS_charac}, we conclude that $M$ is not an NMDS matrix.
	\qed
\end{proof}

\begin{remark}
	It is important to note that when $n$ is odd, a $2n \times 2n$ Type-II circulant-like matrix over $\mathbb{F}_{2^r}$ can be an NMDS matrix. For example, consider the $6 \times 6$ matrix $M = \text{TypeII}(Circ(1, \alpha, \alpha))$ over the finite field $\mathbb{F}_{2^4}$, where $\alpha$ is a primitive element and a root of the constructing polynomial $x^4 + x + 1$. It can be checked that the matrix $M$ is an NMDS matrix.
\end{remark}


\section{Conclusion}\label{Sec:Conclusion}
This paper focuses on two primary objectives. First, it establishes explicit formulas for counting Hadamard MDS and involutory Hadamard MDS matrices of order $4$ over the finite field $\mathbb{F}_{2^r}$, followed by the study of Near-MDS (NMDS) matrices derived from Hadamard and circulant-like matrices. It demonstrates that a Type-II circulant-like matrix of order $2n$ can never be NMDS when $n$ is even, and a Hadamard matrix must be nonsingular to be an NMDS matrix. It also proves that there exist only two orthogonal Type-I circulant-like NMDS matrices over $\mathbb{F}_{2^r}$. Additionally, the paper derives a formula for counting $4 \times 4$ Hadamard NMDS and involutory Hadamard NMDS matrices with exactly one zero per row or column over $\mathbb{F}_{2^r}$. It also presents a general form for constructing $2n \times 2n$ involutory MDS matrices and uses this to compute the number of both order $2$ MDS and order $2$ involutory MDS matrices in $\mathbb{F}_{2^r}$. Based on these findings, an upper bound is proposed for the total number of involutory MDS matrices of order $4$ over $\mathbb{F}_{2^r}$. However, this bound is very loose, leaving room for future work to either refine this estimate or determine the precise count of $4 \times 4$ involutory MDS matrices.

\subsubsection*{Acknowledgments} 
We would like to thank Lukas Koelsch for providing a valuable insight that contributed toward the proof of Theorem~\ref{Thm_Inv_Hadamard_MDS_Counting}. We are also grateful to Sandip Kumar Mondal for his valuable discussions on various aspects of an earlier version of this paper.

\medskip

\bibliographystyle{plain}
\bibliography{Bibliography}

\clearpage

\appendix  
\section{The $4\times 4$ Hadamard MDS matrices over $\mathbb{F}_{2^r}$}\label{Sec:Appendix}

\renewcommand{\arraystretch}{1.5} 
\begin{table}[h!]
	\centering
	\begin{adjustbox}{max width=\textwidth}
	\begin{tabular}{|c|c|c|c|c|c|c|c|}
	\hline
	$(1, \alpha, \alpha^2, \alpha^5)$ & $(1, \alpha, \alpha^4, \alpha^6)$ & $(1, \alpha, \alpha^5, \alpha^2)$ & $(1, \alpha, \alpha^6, \alpha^4)$ & $(1, \alpha^2, \alpha, \alpha^5)$ & $(1, \alpha^2, \alpha^3, \alpha^4)$ & $(1, \alpha^2, \alpha^4, \alpha^3)$ & $(1, \alpha^2, \alpha^5, \alpha)$ \\
	$(1, \alpha^3, \alpha^2, \alpha^4)$ & $(1, \alpha^3, \alpha^4, \alpha^2)$ & $(1, \alpha^3, \alpha^5, \alpha^6)$ & $(1, \alpha^3, \alpha^6, \alpha^5)$ & $(1, \alpha^4, \alpha, \alpha^6)$ & $(1, \alpha^4, \alpha^2, \alpha^3)$ & $(1, \alpha^4, \alpha^3, \alpha^2)$ & $(1, \alpha^4, \alpha^6, \alpha)$ \\
	$(1, \alpha^5, \alpha, \alpha^2)$ & $(1, \alpha^5, \alpha^2, \alpha)$ & $(1, \alpha^5, \alpha^3, \alpha^6)$ & $(1, \alpha^5, \alpha^6, \alpha^3)$ & $(1, \alpha^6, \alpha, \alpha^4)$ & $(1, \alpha^6, \alpha^3, \alpha^5)$ & $(1, \alpha^6, \alpha^4, \alpha)$ & $(1, \alpha^6, \alpha^5, \alpha^3)$ \\
	\hline
	$(\alpha, 1, \alpha^2, \alpha^5)$ & $(\alpha, 1, \alpha^4, \alpha^6)$ & $(\alpha, 1, \alpha^5, \alpha^2)$ & $(\alpha, 1, \alpha^6, \alpha^4)$ & $(\alpha, \alpha^2, 1, \alpha^5)$ & $(\alpha, \alpha^2, \alpha^3, \alpha^6)$ & $(\alpha, \alpha^2, \alpha^5, 1)$ & $(\alpha, \alpha^2, \alpha^6, \alpha^3)$ \\
	$(\alpha, \alpha^3, \alpha^2, \alpha^6)$ & $(\alpha, \alpha^3, \alpha^4, \alpha^5)$ & $(\alpha, \alpha^3, \alpha^5, \alpha^4)$ & $(\alpha, \alpha^3, \alpha^6, \alpha^2)$ & $(\alpha, \alpha^4, 1, \alpha^6)$ & $(\alpha, \alpha^4, \alpha^3, \alpha^5)$ & $(\alpha, \alpha^4, \alpha^5, \alpha^3)$ & $(\alpha, \alpha^4, \alpha^6, 1)$ \\
	$(\alpha, \alpha^5, 1, \alpha^2)$ & $(\alpha, \alpha^5, \alpha^2, 1)$ & $(\alpha, \alpha^5, \alpha^3, \alpha^4)$ & $(\alpha, \alpha^5, \alpha^4, \alpha^3)$ & $(\alpha, \alpha^6, 1, \alpha^4)$ & $(\alpha, \alpha^6, \alpha^2, \alpha^3)$ & $(\alpha, \alpha^6, \alpha^3, \alpha^2)$ & $(\alpha, \alpha^6, \alpha^4, 1)$ \\
	\hline
	$(\alpha^2, 1, \alpha, \alpha^5)$ & $(\alpha^2, 1, \alpha^3, \alpha^4)$ & $(\alpha^2, 1, \alpha^4, \alpha^3)$ & $(\alpha^2, 1, \alpha^5, \alpha)$ & $(\alpha^2, \alpha, 1, \alpha^5)$ & $(\alpha^2, \alpha, \alpha^3, \alpha^6)$ & $(\alpha^2, \alpha, \alpha^5, 1)$ & $(\alpha^2, \alpha, \alpha^6, \alpha^3)$ \\
	$(\alpha^2, \alpha^3, 1, \alpha^4)$ & $(\alpha^2, \alpha^3, \alpha, \alpha^6)$ & $(\alpha^2, \alpha^3, \alpha^4, 1)$ & $(\alpha^2, \alpha^3, \alpha^6, \alpha)$ & $(\alpha^2, \alpha^4, 1, \alpha^3)$ & $(\alpha^2, \alpha^4, \alpha^3, 1)$ & $(\alpha^2, \alpha^4, \alpha^5, \alpha^6)$ & $(\alpha^2, \alpha^4, \alpha^6, \alpha^5)$ \\
	$(\alpha^2, \alpha^5, 1, \alpha)$ & $(\alpha^2, \alpha^5, \alpha, 1)$ & $(\alpha^2, \alpha^5, \alpha^4, \alpha^6)$ & $(\alpha^2, \alpha^5, \alpha^6, \alpha^4)$ & $(\alpha^2, \alpha^6, \alpha, \alpha^3)$ & $(\alpha^2, \alpha^6, \alpha^3, \alpha)$ & $(\alpha^2, \alpha^6, \alpha^4, \alpha^5)$ & $(\alpha^2, \alpha^6, \alpha^5, \alpha^4)$ \\
	\hline
	$(\alpha^3, 1, \alpha^2, \alpha^4)$ & $(\alpha^3, 1, \alpha^4, \alpha^2)$ & $(\alpha^3, 1, \alpha^5, \alpha^6)$ & $(\alpha^3, 1, \alpha^6, \alpha^5)$ & $(\alpha^3, \alpha, \alpha^2, \alpha^6)$ & $(\alpha^3, \alpha, \alpha^4, \alpha^5)$ & $(\alpha^3, \alpha, \alpha^5, \alpha^4)$ & $(\alpha^3, \alpha, \alpha^6, \alpha^2)$ \\
	$(\alpha^3, \alpha^2, 1, \alpha^4)$ & $(\alpha^3, \alpha^2, \alpha, \alpha^6)$ & $(\alpha^3, \alpha^2, \alpha^4, 1)$ & $(\alpha^3, \alpha^2, \alpha^6, \alpha)$ & $(\alpha^3, \alpha^4, 1, \alpha^2)$ & $(\alpha^3, \alpha^4, \alpha, \alpha^5)$ & $(\alpha^3, \alpha^4, \alpha^2, 1)$ & $(\alpha^3, \alpha^4, \alpha^5, \alpha)$ \\
	$(\alpha^3, \alpha^5, 1, \alpha^6)$ & $(\alpha^3, \alpha^5, \alpha, \alpha^4)$ & $(\alpha^3, \alpha^5, \alpha^4, \alpha)$ & $(\alpha^3, \alpha^5, \alpha^6, 1)$ & $(\alpha^3, \alpha^6, 1, \alpha^5)$ & $(\alpha^3, \alpha^6, \alpha, \alpha^2)$ & $(\alpha^3, \alpha^6, \alpha^2, \alpha)$ & $(\alpha^3, \alpha^6, \alpha^5, 1)$ \\
	\hline
	$(\alpha^4, 1, \alpha, \alpha^6)$ & $(\alpha^4, 1, \alpha^2, \alpha^3)$ & $(\alpha^4, 1, \alpha^3, \alpha^2)$ & $(\alpha^4, 1, \alpha^6, \alpha)$ & $(\alpha^4, \alpha, 1, \alpha^6)$ & $(\alpha^4, \alpha, \alpha^3, \alpha^5)$ & $(\alpha^4, \alpha, \alpha^5, \alpha^3)$ & $(\alpha^4, \alpha, \alpha^6, 1)$ \\
	$(\alpha^4, \alpha^2, 1, \alpha^3)$ & $(\alpha^4, \alpha^2, \alpha^3, 1)$ & $(\alpha^4, \alpha^2, \alpha^5, \alpha^6)$ & $(\alpha^4, \alpha^2, \alpha^6, \alpha^5)$ & $(\alpha^4, \alpha^3, 1, \alpha^2)$ & $(\alpha^4, \alpha^3, \alpha, \alpha^5)$ & $(\alpha^4, \alpha^3, \alpha^2, 1)$ & $(\alpha^4, \alpha^3, \alpha^5, \alpha)$ \\
	$(\alpha^4, \alpha^5, \alpha, \alpha^3)$ & $(\alpha^4, \alpha^5, \alpha^2, \alpha^6)$ & $(\alpha^4, \alpha^5, \alpha^3, \alpha)$ & $(\alpha^4, \alpha^5, \alpha^6, \alpha^2)$ & $(\alpha^4, \alpha^6, 1, \alpha)$ & $(\alpha^4, \alpha^6, \alpha, 1)$ & $(\alpha^4, \alpha^6, \alpha^2, \alpha^5)$ & $(\alpha^4, \alpha^6, \alpha^5, \alpha^2)$ \\
	\hline
	$(\alpha^5, 1, \alpha, \alpha^2)$ & $(\alpha^5, 1, \alpha^2, \alpha)$ & $(\alpha^5, 1, \alpha^3, \alpha^6)$ & $(\alpha^5, 1, \alpha^6, \alpha^3)$ & $(\alpha^5, \alpha, 1, \alpha^2)$ & $(\alpha^5, \alpha, \alpha^2, 1)$ & $(\alpha^5, \alpha, \alpha^3, \alpha^4)$ & $(\alpha^5, \alpha, \alpha^4, \alpha^3)$ \\
	$(\alpha^5, \alpha^2, 1, \alpha)$ & $(\alpha^5, \alpha^2, \alpha, 1)$ & $(\alpha^5, \alpha^2, \alpha^4, \alpha^6)$ & $(\alpha^5, \alpha^2, \alpha^6, \alpha^4)$ & $(\alpha^5, \alpha^3, 1, \alpha^6)$ & $(\alpha^5, \alpha^3, \alpha, \alpha^4)$ & $(\alpha^5, \alpha^3, \alpha^4, \alpha)$ & $(\alpha^5, \alpha^3, \alpha^6, 1)$ \\
	$(\alpha^5, \alpha^4, \alpha, \alpha^3)$ & $(\alpha^5, \alpha^4, \alpha^2, \alpha^6)$ & $(\alpha^5, \alpha^4, \alpha^3, \alpha)$ & $(\alpha^5, \alpha^4, \alpha^6, \alpha^2)$ & $(\alpha^5, \alpha^6, 1, \alpha^3)$ & $(\alpha^5, \alpha^6, \alpha^2, \alpha^4)$ & $(\alpha^5, \alpha^6, \alpha^3, 1)$ & $(\alpha^5, \alpha^6, \alpha^4, \alpha^2)$ \\
	\hline
	$(\alpha^6, 1, \alpha, \alpha^4)$ & $(\alpha^6, 1, \alpha^3, \alpha^5)$ & $(\alpha^6, 1, \alpha^4, \alpha)$ & $(\alpha^6, 1, \alpha^5, \alpha^3)$ & $(\alpha^6, \alpha, 1, \alpha^4)$ & $(\alpha^6, \alpha, \alpha^2, \alpha^3)$ & $(\alpha^6, \alpha, \alpha^3, \alpha^2)$ & $(\alpha^6, \alpha, \alpha^4, 1)$ \\
	$(\alpha^6, \alpha^2, \alpha, \alpha^3)$ & $(\alpha^6, \alpha^2, \alpha^3, \alpha)$ & $(\alpha^6, \alpha^2, \alpha^4, \alpha^5)$ & $(\alpha^6, \alpha^2, \alpha^5, \alpha^4)$ & $(\alpha^6, \alpha^3, 1, \alpha^5)$ & $(\alpha^6, \alpha^3, \alpha, \alpha^2)$ & $(\alpha^6, \alpha^3, \alpha^2, \alpha)$ & $(\alpha^6, \alpha^3, \alpha^5, 1)$ \\
	$(\alpha^6, \alpha^4, 1, \alpha)$ & $(\alpha^6, \alpha^4, \alpha, 1)$ & $(\alpha^6, \alpha^4, \alpha^2, \alpha^5)$ & $(\alpha^6, \alpha^4, \alpha^5, \alpha^2)$ & $(\alpha^6, \alpha^5, 1, \alpha^3)$ & $(\alpha^6, \alpha^5, \alpha^2, \alpha^4)$ & $(\alpha^6, \alpha^5, \alpha^3, 1)$ & $(\alpha^6, \alpha^5, \alpha^4, \alpha^2)$ \\
	\hline
	\end{tabular}
	\end{adjustbox}
	\caption{List of $4\times 4$ Hadamard MDS matrices over $\mathbb{F}_{2^3}$. Each tuple represents the first row of a Hadamard matrix, where $\alpha$ denotes a primitive element of $\mathbb{F}_{2^3}$ and is a root of the polynomial $x^3 + x^2 + 1$.}
\end{table}

\begin{table}[h!]
	\centering
	\begin{adjustbox}{max width=\textwidth}
	\begin{tabular}{|c|c|c|c|c|c|c|c|}
	\hline
	$(1, \alpha^3, \alpha^5, \alpha^6)$ & $(1, \alpha^3, \alpha^6, \alpha^5)$ & $(1, \alpha^5, \alpha^3, \alpha^6)$ & $(1, \alpha^5, \alpha^6, \alpha^3)$ & $(1, \alpha^6, \alpha^3, \alpha^5)$ & $(1, \alpha^6, \alpha^5, \alpha^3)$ & $(\alpha^3, 1, \alpha^5, \alpha^6)$ & $(\alpha^3, 1, \alpha^6, \alpha^5)$ \\
	\hline
	$(\alpha^3, \alpha^5, 1, \alpha^6)$ & $(\alpha^3, \alpha^5, \alpha^6, 1)$ & $(\alpha^3, \alpha^6, 1, \alpha^5)$ & $(\alpha^3, \alpha^6, \alpha^5, 1)$ & $(\alpha^5, 1, \alpha^3, \alpha^6)$ & $(\alpha^5, 1, \alpha^6, \alpha^3)$ & $(\alpha^5, \alpha^3, 1, \alpha^6)$ & $(\alpha^5, \alpha^3, \alpha^6, 1)$ \\
	\hline
	$(\alpha^5, \alpha^6, 1, \alpha^3)$ & $(\alpha^5, \alpha^6, \alpha^3, 1)$ & $(\alpha^6, 1, \alpha^3, \alpha^5)$ & $(\alpha^6, 1, \alpha^5, \alpha^3)$ & $(\alpha^6, \alpha^3, 1, \alpha^5)$ & $(\alpha^6, \alpha^3, \alpha^5, 1)$ & $(\alpha^6,\alpha^5,1,\alpha^3)$ & $(\alpha^6,\alpha^5,\alpha^3,1)$ \\
	\hline
	\end{tabular}
	\end{adjustbox}
	\caption{List of $4\times 4$ involutory Hadamard MDS matrices over $\mathbb{F}_{2^3}$. Each tuple represents the first row of a Hadamard matrix, where $\alpha$ denotes a primitive element of $\mathbb{F}_{2^3}$ and is a root of the polynomial $x^3 + x^2 + 1$.}
\end{table}

\end{document}